\documentclass[runningheads]{llncs}
\usepackage[utf8]{inputenc}
\usepackage{multicol}

\usepackage{algorithm}

\usepackage{enumitem}

\usepackage{xspace}
\usepackage{wrapfig}
\usepackage{graphicx}
\usepackage{amsfonts} 
\usepackage{stmaryrd} 
\DeclareMathAlphabet{\mathpzc}{OT1}{pzc}{m}{it}
\usepackage[skip=3pt,font=footnotesize]{caption}

\usepackage[dvipsnames]{xcolor}
\usepackage{tikz}
\usepackage{tablefootnote}
\usepackage{makecell}
\usepackage{multicol}

\newcommand{\commentout}[1]{}

\newcommand{\hadar}[1]{\colorbox{green}{......}\footnote{\colorbox{green}{\textsc{hf:}}~#1\colorbox{green}{.}}}

\renewcommand{\paragraph}[1]{\vspace{1mm}\noindent\textbf{\textit{#1}}\ }

\newcommand{\dmin}{\ensuremath{d_{\textit{inf}}}}
\newcommand{\dmax}{\ensuremath{d_{\textit{sup}}}}

\newcommand{\class}[1]{\ensuremath{{\mathbb{#1}}}}

\newcommand{\lang}[1]{\mathcal{L}({#1})}

\newcommand{\sat}[1]{\ensuremath{\mathit{sat}^{#1}}}
\newcommand{\size}[1]{\ensuremath{\mathit{size}^{#1}}}

\newcommand{\alge}[1]{\ensuremath{\mathpzc{#1}}}
\newcommand{\aut}[1]{\ensuremath{\mathcal{#1}}}
\newcommand{\dom}[1]{\ensuremath{\mathbb{#1}}}
\newcommand{\atpreds}[1]{\ensuremath{\mathbb{#1}_0}}
\newcommand{\sema}[1]{\ensuremath{\llbracket{#1}\rrbracket}}

\newcommand{\sfa}{\textsc{SFA}}

\renewcommand{\dag}{\textsc{dag}}

\title{On the Complexity of \\
Symbolic Finite-State Automata}

\author{Dana Fisman\inst{1} \and
	Hadar Frenkel\inst{2}\thanks{The work of this author was partially supported by the Technion Hiroshi Fujiwara Cyber Security Research Center and the Israel National Cyber Directorate}  \and
	Sandra Zilles\inst{3}  }
\institute{Ben-Gurion University, Be’er Sheva, Israel
\and
 The Technion, Haifa, Israel
 \and
 University of Regina, Regina, Canada}

\begin{document}

\maketitle
\noindent
\makebox[\linewidth]{\small July 2021}

\begin{abstract}
    We revisit the complexity of procedures on SFAs (such as intersection, emptiness, etc.)
    and analyze them according to the measures we find suitable for symbolic automata: the number
    of states ($n$), the maximal number of transitions exiting a state ($m$) and the size
    of the most complex transition predicate ($l$).
    We pay attention to the special forms of SFAs: \emph{normalized SFAs} and \emph{neat SFAs},
    as well as to SFAs over a \emph{monotonic} effective Boolean algebra.
\end{abstract}

\section{Introduction}
\emph{Symbolic finite state automata}, SFAs for short, are an automata model in which transitions between
states correspond to predicates over a domain of concrete alphabet letters. Their purpose is to
cope with situations where the domain of concrete alphabet letters is large or infinite. 
SFAs have proven useful in many applications~\cite{DAntoniVLM14,PredaGLM15,ArgyrosSJKK16,HuD17,SaarikiviV17,MamourasRAIK17} and consequently have been studied as a theoretical model of automata.
Many algorithms for natural questions over these automata already exist in the literature, in particular, Boolean operations, determinization, and  emptiness~\cite{DBLP:conf/icst/VeanesHT10}; minimization~\cite{DAntoniV16}; and language inclusion~\cite{KeilT14}. 

The literature on SFAs has mainly focused on a special type of SFA, termed \emph{normalized}, in which there
is at most one transition between every pair of states. This minimization of the number of transitions comes at the cost
of obtaining more complex predicates. 
We promote another special type of SFA, that we
term a \emph{neat SFA}, which by contrast, allows several transitions between the same pair of states, but restricts
the predicates to be \emph{basic}, as formally defined in Section~\ref{sec:bool-alge}.

Previous results regarding the complexity of automata algorithms on SFAs, have mainly considered the number of states in the SFA as a parameter to the complexity of the algorithm. 
Indeed the other components in a DFA are of fixed size or polynomial in the number of states. In SFAs however, the size of the alphabet, which is a  set of  predicates, can grow unboundedly, and so can the number of transitions. 
Therefore we propose to measure the size of an SFA with respect to three parameters: the number of states ($n$), the maximal out-degree of a state ($m$) and the size of the most complex predicate ($l$). We revisit the results in the literature for the complexity of standard automata operations, and analyze them along the proposed measures. 
In addition, we show how the complexity of these operations vary according to the special form. 
We show that most procedures are more efficient on neat SFAs.
In addition, we show that  SFAs over a monotonic algebra have  a unique minimal neat SFA and a canonical minimal normalized SFA; 
and the transformation between the different forms is at most polynomial in these three parameters,

\section{Preliminaries} \label{sec:SFA_prelim}

\emph{Symbolic finite-state automata}, shortened as \emph{symbolic automata}, abbreviated as \emph{SFA} are defined
with respect to an effective Boolean Algebra. We thus start with defining effective Boolean Algebras \S\ref{sec:bool-alge},
then provide the definition of SFAs\S\ref{sec:symbolic_automata}

\subsection{Effective Boolean Algebra}\label{sec:bool-alge}
\emph{A Boolean Algebra} \alge{A} is a tuple $\langle \dom{D}, \class{P}, \sema{\cdot}, \bot, \top, \vee, $ 
$\wedge, \neg\rangle$ where $\dom{D}$ is a set of domain elements; 
$\class{P}$ is a set of predicates closed under the Boolean connectives, where $\bot,\top\in\class{P}$; the component $\sema{\cdot} : \class{P}\rightarrow 2^\dom{D}$ is the so-called \emph{semantics function}. $\dom{P}$ satisfies the following three requirements: 
(i) $\sema{\bot} = \emptyset$, 
(ii) $\sema{\top} = \dom{D}$,~~and 
(iii) for all $\varphi,\psi\in \class{P}$, $~~\sema{\varphi\vee \psi} = \sema{\varphi}\cup\sema{\psi}$, $~~\sema{\varphi\wedge \psi} = \sema{\varphi}\cap\sema{\psi}$, and $~~\sema{\neg\varphi}= \dom{D}\setminus\sema{\varphi}$. A Boolean Algebra is \emph{effective} if all the operations above, as well as satisfiability, are decidable.
Henceforth, we implicitly assume Boolean algebras to be effective.

One way to define a Boolean algebra is by defining a set $\class{P}_0$ of \emph{atomic formulas} and obtaining $\class{P}$ by closing $\class{P}_0$ for conjunction, disjunction and negation.
For a predicate $\psi\in\class{P}$ we say that $\psi$ is \emph{atomic} if $\psi\in\class{P}_0$ or $\psi\in\{\top,\bot\}$. 
We say that 
$\psi$ is \emph{basic} if $\psi$ is a conjunction of atomic formulas.

\begin{example}
\label{ex:interval-algebra}
	The \emph{{interval algebra}}
	is the Boolean algebra in which the domain $\dom{D}$ is the set  $\dom{Z}\cup\{-\infty,\infty\}$ of integers augmented with two special symbols with their standard semantics, and the set of atomic formulas $\class{P}_0$ consists of  intervals of the form $[a,b)$ where $a,b\in\dom{D}$ and $a<b$. The semantics associated with intervals is the natural one: $\sema{[a,b)}=\{z\in\dom{D}~:~a\leq z \mbox{ and } z < b\}$.

\end{example}

\begin{example}
\label{ex:propositional-algebra}
	The \emph{propositional algebra}
	is defined with respect to a set $AP=\{p_1,p_2,\ldots,p_k\}$ of atomic propositions. 
	The set of \emph{atomic predicates} $\atpreds{P}$ consists of the atomic propositions and their negations.
	The  domain $\dom{D}$ consists of all the possible valuations for these propositions, thus $\dom{D} = \dom{B}^k$ where $\dom{B}=\{0,1\}$. The semantics of an atomic predicate $p$ is given by $\sema{p_i}=\{v\in\dom{B}^k~:~v[i]=1\}$, and similarly $\sema{\neg p_i}=\{v\in\dom{B}^k~:~v[i]=0\}$. In this case  a basic formula is a \emph{monomial}, that is, a conjunction of atomic predicated and their negations.
\end{example}

\subsection{Symbolic Automata} \label{sec:symbolic_automata}
A \emph{symbolic finite-state automaton} (\sfa) is a tuple 
$\aut{M}=
\langle \alge{A}, Q, q_0, \delta, F   \rangle
$ where $\alge{A}$ 
is a Boolean algebra, 
$Q$ is a finite set of states, $q_0 \in Q$ is the initial state, 
$F \subseteq Q$ is the set of final states, and $\delta \subseteq Q \times \class{P}_\alge{A} \times Q $ is the transition relation, where $\class{P}_\alge{A} $ is the set of predicates of $\alge{A}$.

We use the term \emph{letters} 
for elements of $\dom{D}$ where $\dom{D}$ is the domain of $\alge{A}$ and the term \emph{words} 

for elements of $\dom{D}^*$.
A run of $\aut{M}$ on a word $\sigma_1 \sigma_2\ldots \sigma_n$ where $\sigma_i\in \dom{D}$, is a sequence of transitions
 $\langle q_0,\psi_1,q_1\rangle \langle q_1,\psi_2,q_2\rangle\ldots\langle q_{n-1},\psi_n,q_n\rangle$ satisfying that 
 $\sigma_i\in \sema{\psi_i}$ and that 
 $\langle q_i,\psi_{i+1},q_{i+1}\rangle\in\delta$. 
 Such a run is said to be \emph{accepting} if $q_n\in F$. A word $w=\sigma_1 \sigma_2\ldots \sigma_n$  is said to be \emph{accepted} 
  by $\aut{M}$ if there exists an accepting run of $\aut{M}$ on $w$. The set of words accepted by an SFA $\aut{M}$ is denoted 
$\lang{\aut{M}}$.

An \sfa\ is said to be \emph{deterministic} if for every state $q\in Q$ and every letter $\sigma\in\dom{D}$ we have that
$|\{\langle q,\psi,q' \rangle \in\delta~:~\sigma\in\sema{\psi}\}|\leq 1$, namely from every state and every concrete letter there exists at most one transition. 
It is said to be \emph{complete} if  $|\{\langle q,\psi,q'\rangle\in\delta~:~\sigma\in\sema{\psi}\}|\geq 1$ for every  $q\in Q$ and every  $\sigma\in\dom{D}$, namely 
 from every state and every concrete letter there exists at least one transition. 
As is the case for finite automata (over concrete alphabets), non-determinism does not add expressive power
but does add succinctness~\cite{DBLP:conf/icst/VeanesHT10}.

\section{Types of Symbolic Automata}\label{sec:types}

We turn to define special types of \sfa s, which affect the complexity of related procedures. 

\quad \\
\paragraph{Neat and Normalized SFAs}\quad \\
We note that there is a trade-off between the number of transitions, and the
complexity of the transition predicates.
The literature defines an \sfa\ as \emph{normalized} if for every two states $q$ and $q'$ there exists at most one transition from $q$ to $q'$.  
This definition prefers fewer transitions at the cost of potentially complicated predicates.
By contrast, preferring simple transitions at the cost of increasing the number of transitions, leads to~\emph{neat} \sfa s.
We define an \sfa\ to be \emph{neat} if
all transition predicates are basic predicates. 

\quad \\
\paragraph{Feasibility}\quad \\
The second distinction concerns the fact that an \sfa\ can have transitions with unsatisfiable predicates.
A symbolic automaton is said to be \emph{feasible}  if for every $\langle q,\psi,q'\rangle\in\delta$ we have that $\sema{\psi}\neq \emptyset$. 
{Feasibility is an orthogonal property to being neat or normalized.} 

\quad \\
\paragraph{Monotonicity}\quad \\
The third distinction we make concerning the nature of a given \sfa\ regards its underlying algebra. 
A Boolean algebra $\alge{A} $ over domain $\dom{D}$ is said to be  \emph{monotonic} if
the following hold.
\begin{enumerate}
    \item There exists a total order $<$ on the elements of $\dom{D}$;
    and
    \item There exist two elements $\dmin$ and $\dmax$ such that $\dmin\leq d$ and $d \leq \dmax$ for all $d\in\dom{D}$;
    and
    \item An atomic predicate $\psi\in\atpreds{P}$ can be associated with two concrete values $a$ and $b$ such that $\sema{\psi}=\{d\in\dom{D}~:~ a\leq d < b \}$.
\end{enumerate}

\noindent The interval algebra of Example~\ref{ex:interval-algebra} is clearly monotonic, as is the similar algebra obtained using $\dom{R}$ (the real numbers) instead of $\dom{Z}$ (the integers). On the other hand, the propositional algebra of Example~\ref{ex:propositional-algebra} is clearly non-monotonic. 
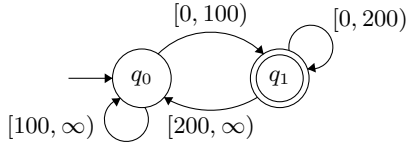
\begin{figure}[t]
	\centering
\begin{tikzpicture}[scale=0.13]
\tikzstyle{every node}+=[inner sep=0pt]
\draw [black] (18.3,-19) circle (3);
\draw (18.3,-19) node {$q_0$};
\draw [black] (32.4,-19) circle (3);
\draw (32.4,-19) node {$q_1$};
\draw [black] (32.4,-19) circle (2.4);
\draw [black] (10.8,-19) -- (15.3,-19);
\fill [black] (15.3,-19) -- (14.5,-18.5) -- (14.5,-19.5);
\draw [black] (19.952,-16.519) arc (135.06693:44.93307:7.625);
\fill [black] (30.75,-16.52) -- (30.54,-15.6) -- (29.83,-16.31);
\draw (25.35,-13.78) node [above] {$[0,100)$};
\draw [black] (33.36,-16.17) arc (189:-99:2.25);
\draw (37.7,-13) node [right] {$[0,200)$};
\fill [black] (35.23,-18.04) -- (36.1,-18.41) -- (35.94,-17.42);
\draw [black] (18.597,-21.973) arc (33.44395:-254.55605:2.25);
\draw (9,-22.77) node [below] {$[100,\infty)$};
\fill [black] (16.12,-21.04) -- (15.18,-21.07) -- (15.73,-21.9);
\draw [black] (30.128,-20.939) arc (-58.82797:-121.17203:9.231);
\fill [black] (20.57,-20.94) -- (21,-21.78) -- (21.52,-20.93);
\draw (25.35,-22.77) node [below] {$[200,\infty)$};
\end{tikzpicture}
\caption{{The SFA $\aut{M}$ over the interval algebra}}
\label{fig:SFA}
\end{figure}

\begin{example}\label{ex:SFA}
Consider the SFA $\aut{M}$ given in Fig~\ref{fig:SFA}. It is defined over the algebra ${\alge{A_\dom{N}}}$ which is the interval algebra restricted to the domain $\dom{D} = \mathbb{N}\cup \{\infty  \}$. The language of $\aut{M}$ is the set of all words over $\mathbb{N}$ of the form $w_1 \cdot d \cdot w_2$ where $w_1$ is some word over the domain $\dom{D}$, $0\leq d <100$ and $w_2$ consists 
of numbers smaller than $200$. 
$\aut{M}$ is defined over a monotonic algebra, and is neat, normalized, deterministic and complete.  
\end{example}

\subsection{Size of an \sfa} 
The size of an automaton (not a symbolic one) is typically measured by its number of states. This is since for DFAs, the size of the alphabet is assumed to be a given constant, and the rest of the parameters, in particular the transition relation, are at most quadratic in the number of states. In the case of \sfa s the situation is different, as the size of the predicates labeling the transitions can vary greatly. In fact, if we measure the size of a predicate by the number of nodes in its parse \dag, then the size of a formula can grow unboundedly.
The size and structure of the predicates influence the complexity of their satisfiability check, and thus the complexity of corresponding algorithms. On the other hand there might be a trade-off between the size of the transition predicates and the number of transitions; e.g. a predicate of the form $\psi_1 \vee \psi_2 \ldots \vee \psi_k$ can be replaced by $k$ transitions, each one labeled by $\psi_i$ for $1\leq i \leq k$. 

Therefore, we measure the size of an \sfa\ by three parameters: the number of states ($n$), the maximal out-degree of a state ($m$) and the size of the most complex predicate ($l$).
The size of a predicate is defined as the size of its parse tree, i.e. the number of atomic predicates and the number of logic operators occurring in it.
In order to analyze the complexity of automata algorithms discussed in Sections~\ref{sec:boolean_operations} and~\ref{sec:special_boolean_operations}, for a class $\class{P}$ of predicates over a Boolean algebra $\alge{A}$, we 
also use the following measures: the complexity measure 
$\sat{\class{P}}(l)$, which is the complexity of satisfiability check for a predicate of length $l$ in $\class{P}$; and the size measure
$\size{\class{P}}_{\wedge}(l_1, l_2)$ (or $\size{\class{P}}_{\vee}(l_1, l_2)$), which is the size of the conjunction (disjunction) of two predicates in $\class{P}$. 
While for the interval algebra $\size{\class{P}}_{\wedge}(l_1, l_2)$ is linear in $l_1$ and $l_2$, for the OBDD (ordered binary decision diagrams) algebra Boolean operations on predicates are polynomial~\cite{DBLP:journals/tc/Bryant86}. When the algebra is built on a set of atomic predicates $\class{P}_0$ we also use $\sat{\class{P}_0}(l)$, $\size{\class{P}_0}_{\wedge}(l_1, l_2)$ and $\size{\class{P}_0}_{\vee}(l_1, l_2)$, for the respective complexities when restricted to atomic predicates.



\section{Transformations to Special Forms}
We now address the task of transforming SFAs into their special forms as presented in Section~\ref{sec:types}.
We discuss transformations to the special forms  \emph{neat}, \emph{normalized} and \emph{feasible} automata, measured
as suggested using $\langle n, m,l \rangle$ --- the number of states, the maximal out-degree of a state, and the size of the most complex predicate.
\commentout{
\hadar{a bit repetitive so I took this off, but we can return}
First, we note that neat and normalized automata are opposite to one another. Neat automata have concise thus smaller edges, represented by the $l$ parameter. Yet, they might have a larger out-degree $m$, comparing to normalized automata. Normalized automata might have large edges since they include all symbolic transitions between two states, which makes the out degree smaller. 

In addition, we discuss feasible automata, a property which is orthogonal to being neat/normalized. Feasible automata are such that all of their predicates are satisfiable, thus they do not include empty transitions. 
}

\subsection{Neat Automata}\label{sec:neat}
Since each  predicate in a neat SFA is a conjunction of atomic predicates,  neat automata are very intuitive, and the number of transitions in the SFA reflects the complexity of the different operations, as opposed to the situation with normalized SFAs. 
For the class $\class{P}_0$ of basic formulas, 
$\sat{\class{P}_0}(l)$  is usually more efficient than $\sat{\class{P}}(l)$, and in particular is
polynomial for the algebras we consider here. 
This is since satisfiability testing can be reduced to checking that for a basic predicate $\varphi$ that is a conjunction of $l$ atomic predicates, there are no two atomic predicates that contradict each other.
Since satisfiability checking directly affects the complexity of various algorithms discussed in Section~\ref{sec:boolean_operations}, neat SFAs allow for efficient automata operations, as we show in Section~\ref{sec:special_boolean_operations}.

\quad \\
\paragraph{Transforming to Neat} \quad \\
Given a general SFA $\aut{M}$ of size $\langle n, m,l \rangle$, we can construct a neat SFA $\aut{M}'$ of size $\langle n, m\cdot 2^l, l \rangle$, by transforming each transition predicate to a DNF
formula, and turning each disjunct  into an individual transition. 
The number of states, $n$, remains the same. However, the number of transitions can grow exponentially due to the transformation to DNF.
In the worst case, the size of the most complex  predicate can remain the same after the transformation, resulting in the same $l$ parameter for both automata. 
 Note that there is no unique minimal neat SFA. For instance, a predicate $\psi$ over the propositional algebra with $AP=\{p_1,p_2,p_3\}$, satisfying $\sema{\psi}=\{[100],[101],[111]\}$ can be represented using two basic transitions $(p_1\wedge \neg p_2)$ and $(p_1 \wedge p_2 \wedge p_3)$; or alternatively using the two basic transitions: $(p_1\wedge p_3)$ and $(p_1 \wedge \neg p_2 \wedge \neg p_3)$, though it cannot be represented using one basic transition.

Although in the general case, the transformation from normalized to neat SFAs is exponential, for monotonic algebras we have the following lemma, which follows directly from the definition of monotonic algebras and basic predicates.

\begin{lemma}\label{obsrv:basic_monotonic}
Over a monotonic algebra,
the conjunction of two atomic predicates is also an atomic predicate; 
inductively, any basic formula that does not contain negations, over a monotonic algebra, is an atomic predicate.
In addition, the
negation of an atomic predicate is a disjunction of at most 2 atomic predicates.
\end{lemma}

\begin{lemma}\label{lemma:monotonic_to_neat}
Let $\aut{M}$ be a normalized SFA over a monotonic algebra $\alge{A}_{\mathit{mon}}$.
Then, transforming $\aut{M}$ into a neat SFA $\aut{M}'$ is linear in the size of $\aut{M}$.
\end{lemma}

Since a DNF formula with $m$ disjunctions is a natural representation of~$m$~neat transitions, 
Lemma~\ref{lemma:monotonic_to_neat} follows from the following property of monotonic algebras.

\begin{lemma}\label{lemma:dnf_monotonic}
Let $\psi$ be a general formula over a monotonic algebra $\alge{A}_{\mathit{mon}}$. Then,
there exists an equivalent DNF formula $\psi_{\mathit{d}}$ of size linear in $|\psi|$.
\end{lemma}

 \begin{proof}
First, we transform $\psi$ into a Negation Normal Form formula $\psi_{\mathit{NNF}}$, pushing negations inside the formula. 
When transforming to NNF, the number of atomic predicates (possibly under negation) remains the same, and so is the number of conjunctions and disjunctions. 
Since, by Lemma~\ref{obsrv:basic_monotonic}, a negation of an atomic predicate  over a monotonic algebra, namely a negation of an interval, results in at most two intervals, we get that $|\psi_{\mathit{NNF}}| \leq 2\cdot |\psi|$. Note that $\psi_{\mathit{NNF}}$ does not contain any negations, as they were applied to the intervals.
We now transform $\psi_{\mathit{NNF}}$ into a DNF formula $\psi_d$ recursively, operating on sub-formulas of $\psi_{\mathit{NNF}}$, distributing conjunctions over disjunctions.

We inductively prove that $\sema{\psi_d} = \sema{\psi_{\mathit{NNF}}}$ and $|\psi_d| \leq |\psi_{\mathit{NNF}}|$.
For the base case, if $\psi_{\mathit{NNF}}$ is a single interval $[a,b)$, then $[a,b)$ is in DNF and we are done.

For the induction step, consider the two cases.

\begin{enumerate}
    \item Assume $\psi_{\mathit{NNF}} = \psi_1 \vee \psi_2$. By the induction hypothesis, there exists DNF formulas $\psi_{1d}$ and $\psi_{2d}$ such that $
    \sema{\psi_{id}} = \sema{\psi_i}$ and $|\psi_{id}| \leq |\psi_i|$ for $i=1,2$.
    Then, $\psi_d = \psi_{1d}\vee \psi_{2d}$ is equivalent to $\psi_{\mathit{NNF}}$ and of the same size. 
    \item Assume $\psi_{\mathit{NNF}} = \psi_1 \wedge \psi_2$. Again, by the induction hypothesis, instead of $\psi_1\wedge \psi_2$ we can consider $\psi_{1d}\wedge\psi_{2d}$ where $\psi_{1d}$ and $\psi_{2d}$ are in DNF. That is $\psi_{1d} = \bigvee_{i=1}^k [a_i, b_i)$ and $\psi_{2d} = \bigvee_{j=1}^l[c_j, d_j)$. 
    Now, 
    
    $$\psi_{1d} \wedge\psi_{2d} = \left( \bigvee_{i=1}^k [a_i, b_i) \right) \wedge \left( \bigvee_{j=1}^l[c_j, d_j) \right) = \bigvee_{i=1}^{k}\bigvee_{j=1}^{l} \Big(  [a_i, b_i)\wedge[c_j, d_j) \Big)$$
    
    From properties of intervals, each conjunction $[a_i, b_i)\wedge[c_j, d_j)$ is of the form
 $[\max\{a_i, c_j\},\allowbreak \min\{b_i, d_j\})$. The intervals in $\{ [a_i, b_i) ~:~ 1\leq i\leq k \}$ do not intersect (otherwise it would have resulted in a longer single interval), and the same for $\{[c_j, d_j) ~:~ 1\leq j\leq l\}$. Thus, every element $a_i$ or $c_j$ can define at most one interval of the form $[\max\{a_i, c_j\}, \min\{b_i, d_j\})$. That is, the DNF formula $\psi_d = \bigvee_{i=1}^{k}\bigvee_{j=1}^{l} \Big([a_i, b_i)\wedge[c_j, d_j)\Big)$ contains at most $k+l$ intervals, as the others are not proper intervals. Since the size of the original $\psi_{NNF}$ is $k+l$, we have that $|\psi_d|\leq|\psi_{NNF}|$.
\end{enumerate}
To conclude, since $\psi_{NNF}$ is linear in the size of $\psi$ and $\psi_d$ is of the same size as $\psi_{NNF}$, we have that the translation of $\psi$ into the DNF formula $\psi_d$ is linear.  \qed
 \end{proof}

\subsection{Normalized Automata}
Neat automata stand in contrast to normalized ones. In a normalized SFA, there is at most one transition between every pair of states, which  allows for a succinct formulation of the condition to transit from one state to another.
On the other hand, this makes the predicates on the transitions {structurally} more complicated. Given a general SFA $\aut{M}$ with parameters $\langle n, m,l \rangle$, we can easily construct a normalized SFA $\aut{M}'$ as follows. For every pair of states $q$ and $q'$, construct a single edge labeled with the predicate $\bigvee _{\langle q, \varphi, q' \rangle \in\delta} \varphi$. Then, $\aut{M}'$ has size  $\langle n , n,  \size{\class{P}}_{\vee^m}(l)  \rangle$, where we use $\size{\class{P}}_{\vee^m}(l)$ to denote the size of $m$ disjunctions of predicates of size at most $l$. Note that there is no unique minimal normalized automaton either, since in general Boolean formulas have multiple representations.
 However, in Section~\ref{sec:special_boolean_operations} we show that over monotonic algebras there is a  canonical minimal normalized SFA.

The complexity of $\sat{\class{P}}(l)$ for general formulas (corresponding to normalized SFAs) is usually exponentially higher than for basic predicates (and thus for neat SFAs). 
In addition, as we saw above, generating a normalized automaton is an easy operation. This motivates working with neat automata, and generating normalized automata as a last step, if desired (e.g., for presenting a graphical depiction of the automaton). 

\subsection{Feasible Automata}\label{sec:feasible_automata}
The motivation for feasible automata is clear; if the automaton contains unsatisfiable transitions, then its size is larger than necessary, and the redundancy of transitions makes it less interpretable. Thus, infeasible \sfa s add complexity both algorithmically and for the user, as they are more difficult to understand. In order to generate a feasible SFA from a given SFA $\aut{M}$, we need to traverse the transitions of $\aut{M}$ and  test the satisfiability of each transition. 
The parameters $\langle n,m,l   \rangle$ of the SFA remain the same since 
there is no change in the set of states, and there might be no change in transitions as well (if they are all satisfiable). 

In the following, we usually assume that the automata are feasible, and when applying algorithms, we require the output to be feasible as well.

\section{Complexity of standard automata procedures}
In this section we analyze the complexity of automata procedures on SFAs, in terms
of their affect on the parameters $\langle n,m, l\rangle$. We start in \S\ref{sec:boolean_operations} with
examining general SFAs, and then in \S\ref{sec:special_boolean_operations} discuss the affects on special SFAs.

\subsection{Complexity of Automata Procedures for General SFAs}\label{sec:boolean_operations}
We turn to discuss Boolean operations, determinization and minimization, and decision procedures (such as emptiness and equivalence) for the different types of SFAs. For intersection and union, the product construction of SFAs was studied in~\cite{DBLP:conf/icst/VeanesHT10,DBLP:conf/vmcai/HooimeijerV11}. There, the authors assume a normalized SFAs as input, and do not delve on 
the effect of the construction on the number of transitions and the complexity of the resulting predicates. Determinization of SFAs was studied in~\cite{DBLP:conf/icst/VeanesHT10}, and~\cite{DBLP:conf/popl/DAntoniV14} study minimization of SFAs,  
assuming the given SFA is~normalized.

Table~\ref{table:operations} shows the sizes of the SFAs resulting from the mentioned operations, in terms of $\langle n, m,l \rangle$. The analysis applies to all types of SFAs, not just normalized ones.
The time complexity for each operation is given in terms of the parameters $\langle n,m,l  \rangle$ and the complexity of feasibility tests for the resulting SFA, as discussed in Section~\ref{sec:feasible_automata}. 
Table~\ref{table:decision} summarizes the time complexity of decision procedures for SFAs: emptiness, inclusion, and membership. Again, the analysis applies to all types of SFAs.

In both tables we consider two SFAs  $\aut{M}_1$ and $\aut{M}_2$ with parameters $\langle n_i, m_i, l_i \rangle$ for $i=1,2$, over algebra $\alge{A}$ with predicates $\class{P}$. 
We use $\size{\class{P}}_{\wedge^m}(l)$ for an upper bound on the size of $m$ conjunctions of predicates of size at most $l$. 
All SFAs are assumed to be deterministic, except of course for the input for determinization.

\begin{table}[t]
\centering
\begin{tabular}{|c|c|}
\hline
\textbf{Operation} & $\mathbf{\langle n, m, l \rangle } $\\ \hline 
  product construction $\aut{M}_1$, $\aut{M}_2$  &  $\langle n_1 \times n_2,\ m_1 \times m_2,\ \size{\class{P}}_{\wedge}(l_1, l_2) \rangle$

      \\[1mm]
     complementation of deterministic $\aut{M}_1$\tablefootnote{For complementation, no feasibility check is needed, since we assume a feasible input.
     }    &  $\langle n_1+1 ,\ m_1+1,\ \size{\class{P}}_{\vee^{m_1}}(l_1) \rangle$
     \\[1mm]
     determinization of $\aut{M}_1$    &  $\langle 2^{n_1},\ 2^{m_1},\ \size{\class{P}}_{\wedge^{n_1\times m_1}}(l_1)
 \rangle$ \tablefootnote{To determinize transitions, conjunction may be applied $n_1\times m_1$ times, according to the number of states that correspond to a new deterministic state.
 }
 \\[1mm]
  minimization of $\aut{M}_1$    &  $\langle n_1, m_1,\ \size{\class{P}}_{\wedge^{m_1}}(l_1)
 \rangle$
     
     \\ \hline
\end{tabular}
 \quad \\
\caption{Analysis of standard automata procedures on SFAs. 
}
\label{table:operations}
\end{table}

\begin{table}[t]
\centering
\begin{tabular}{|c|c|}  \hline 
\textbf{Decision Procedures} & \textbf{Time Complexity}\\ 
\hline  
  emptiness &  linear in $n, m$  
  \\[1mm]
  emptiness + feasibility &   $n\times m \times \sat{\class{P}}(l)$
  \\[1mm]

  membership of $\gamma_1 \cdots \gamma_t\in\mathbb{D}^*$ & 
  $\sum_{i=1}^t  \sat{\class{P}}(         \size{\class{P}}_{\wedge}(l, |\psi_{\gamma_i}|)) $ \tablefootnote{Where $\psi_{\gamma_i}$ is a predicate describing $\gamma_i$.}
  \\[1mm]
  inclusion $\aut{M}_1\subseteq\aut{M}_2$ & \makecell{
  $((n_1\times n_2) \times (m_1\times m_2 )\times \sat{\class{P}} (\size{\class{P}}_{\wedge}(l_1, l_2)))$ 
  }
  \\[1mm]
 \hline
\end{tabular}
\quad \\
\caption{Analysis of times complexity of decision procedures for SFAs 
}
\label{table:decision}

\end{table}

We now briefly describe the algorithms we analyze in both tables. 

\quad \\
\paragraph{Product Construction~\cite{DBLP:conf/icst/VeanesHT10,DBLP:conf/vmcai/HooimeijerV11}} \quad \\
The product construction for SFAs is similar to the product of DFAs -- the set of states is the product of the states of $\aut{M}_1$ and $\aut{M}_2$; and a transition is a synchronization of transitions of $\aut{M}_1$ and $\aut{M}_2$. That is, a transition from $\langle q_1, q_2 \rangle$ to $\langle p_1, p_2 \rangle$ 
can be made while reading a concrete letter $\gamma$, 
iff $\langle q_1, \psi_1, p_1 \rangle\in \delta_1$ and $\langle q_2, \psi_2, p_2 \rangle\in \delta_2$ and $\gamma$ satisfies both $\psi_1$ and $\psi_2$. Therefore, the predicates labeling transitions in the product construction are conjunctions of predicates from the two SFAs $\aut{M}_1$ and $\aut{M}_2$.  

\quad \\ 
\paragraph{Complementation}\quad \\
In order to complement a deterministic SFA $\aut{M}_1$, we first need to make $\aut{M}_1$ complete. In order to do so, we add one state which is a non-accepting sink, and from each state we add at most one transition which is the negation of all other transitions from that state.   
If $\aut{M}_1$ is complete, then complementation simply switches accepting and non-accepting states, resulting in the same parameters $\langle n_1, m_1, l_1 \rangle$. 

\quad \\
\paragraph{Determinization~\cite{DBLP:conf/icst/VeanesHT10}}\quad \\
In order to make an SFA deterministic, the algorithm of~\cite{DBLP:conf/icst/VeanesHT10} uses the subset construction for DFAs, resulting in an exponential blowup in the number of states. However, in the case of SFAs this is not enough, and the predicates require special care. Let
 $P = \{ q_1, \cdots, q_t \}$ 
 be a state in the deterministic SFA, where $q_1, \ldots, q_t$ are states of the original SFA $\aut{M}_1$, and let $\psi_1, \ldots, \psi_t$ be some predicates labelling outgoing transitions from $q_1, \ldots q_t$, correspondingly. 
 Then, in order to determinize transitions, the algorithm of \cite{DBLP:conf/icst/VeanesHT10} 
computes the conjunction $\bigwedge_{i=1}^t \psi_1$, which labels a single transition from the state $P$. 

\quad \\
\paragraph{Minimization~\cite{DBLP:conf/popl/DAntoniV14}}\quad \\
Given a deterministic SFA $\aut{M}_1$, the output of minimization is an equivalent deterministic SFA with a minimal number of states. When constructing such an SFA, the number of states and transitions cannot grow. However, as in determinization, if two states of $\aut{M}_1$ are replaced with one state, then outgoing transitions might overlap, resulting in a non-deterministic SFA. Therefore, to make sure that transitions do not overlap, all algorithms described in~\cite{DBLP:conf/popl/DAntoniV14} compute \emph{minterms}, which are the smallest conjunctions of outgoing transitions. Minterms then do not intersect, and thus the output is deterministic.   

\quad \\
\paragraph{Emptiness}\quad \\ 
If we assume a feasible SFA $\aut{M}$ as an input, then in order to check for emptiness we need to find an accepting state which is reachable from the initial state (as in DFAs). If we do not assume a feasible input, we need to test the satisfiability of each transition, thus the complexity depends on the complexity measure $\sat{\class{P}}(l)$.

\quad \\
\paragraph{Membership} \quad \\
Similarly to emptiness, in order to check if a concrete word $\gamma_1 \cdots \gamma_n$ is in $\aut{L}(\aut{M})$, we need
not only check if it reaches an accepting state but also
locally consider the satisfiability of each transition. In the case of membership, we need to check whether the letter $\gamma_i$ satisfies the predicate on the corresponding transition. 

\quad \\ \\
\paragraph{Inclusion} \quad \\
Deciding inclusion amounts to checking emptiness and feasibility of $\aut{M}_1  \cap \overline{\aut{M}_2}$. We assume here that both $\aut{M}_1$ and $\aut{M}_2$ are deterministic and complete.

\subsection{Complexity of Automata Procedures for Special SFAs}\label{sec:special_boolean_operations}
We now discuss the advantages of neat SFAs and of monotonic algebras, in the context of the algorithms presented in the tables, and show that, in general, they are more efficient to handle compared to other SFAs.

\subsubsection{Neat SFAs}
 As can be observed from Table~\ref{table:decision}, almost all decision procedures regarding SFAs depend on $\sat{\class{P}}(l)$.
For neat SFAs it is more precise to say that they depend on $\sat{\class{P}_0}(l)$, 
namely on the satisfiability of atomic predicates rather than arbitrary predicated.
Since $\sat{\class{P}_0}(l)$ is usually less costly than $\sat{\class{P}}(l)$, most decision procedures are more efficient 
on neat automata.
Here, we claim that applying automata algorithms on neat SFAs preserves their neatness, thus suggesting that neat SFAs may be preferable in many applications.

\begin{lemma}
Let $\aut{M}_1$ and $\aut{M}_2$ be neat SFAs. Then:  
 $\aut{M}_1 \cap \aut{M}_2$, $\aut{M}_1 \cup \aut{M}_2$, 
 $\overline{\aut{M}_1}$, and determinization / minimization of $\aut{M}_1$, 
are all neat SFAs as well. 
\end{lemma}

\begin{proof}
The proof follows from the product construction~\cite{DBLP:conf/icst/VeanesHT10,DBLP:conf/vmcai/HooimeijerV11}
and the determinization~\cite{DBLP:conf/icst/VeanesHT10} and minimization~\cite{DBLP:conf/popl/DAntoniV14} constructions. All of these use only conjunctions in order to construct the predicates on the output SFAs. Thus, if the predicates on the input SFAs are basic, then so are the output predicates. 
\qed
\end{proof}

\subsubsection{Monotonic Algebras}
We now consider the class $\class{M}_{\alge{A}_{mon}}$ of SFAs over a monotonic algebra $\alge{A}_{mon}$ with predicates $\class{P}$.
We first discuss $\size{\class{P}}_{\wedge}(l_1, l_2)$ and  $\sat{\class{P}}(l)$, as they are essential measures in automata operations.
Then we show that for $\aut{M}_1$ and $\aut{M}_2$ in the class $\class{M}_{\alge{A}_{mon}}$, the product construction is linear in the number of transitions, adding to the efficiency of SFAs over monotonic algebras.

\begin{lemma}\label{lem:conjunction-in-monotonic}
    Let $\psi_1$ and $\psi_2$ be formulas over a monotonic algebra $\alge{A}_{mon}$. 
    Then: $\size{\class{P}}_{\wedge}(|\psi_1|, |\psi_2|)$ is linear in $|\psi_1|+|\psi_2|$ and  $\sat{\class{P}}(|\psi_1|)$ is linear in $|\psi_1|$. 
\end{lemma}

\begin{proof}
Transforming to DNF is linear, as we show in Lemma~\ref{lemma:dnf_monotonic}. There, we showed that the conjunction of two DNF formulas of sizes $k$ and $l$ has size $k+l$, which implies that the conjunction of general formulas has linear size.
In addition, $\sat{\class{P}}(l)$ is trivial for a single interval,
and following Lemma~\ref{lemma:dnf_monotonic}, is linear for general formulas.
The satisfiability of a single interval is trivial, since we define intervals as predicates of the form $[a, b)$ for $a<b$, and thus every interval is satisfiable. Even if we allow unsatisfiable intervals, satisfiability check will amount to the question ``is $a<b$?''. 
\qed
\end{proof}
 
\begin{lemma}\label{lemma:product monotonic}
    Let $\aut{M}_1$ and $\aut{M}_2$ be deterministic SFAs over a monotonic algebra $\alge{A}_{mon}$. Then the out-degree of their product SFA $\aut{M}$ is at most $m = 2\cdot (m_1 + m_2)$.
\end{lemma}
 
 \begin{proof}
From Lemma~\ref{lemma:monotonic_to_neat} and Lemma~\ref{lemma:dnf_monotonic}, 
we can construct neat SFAs $\aut{M}'_1$ and $\aut{M}'_2$ of sizes $\langle n_i, 2m_i,$ $ l_i\rangle$ for $i\in\{1,2\}$.
Similarly to the proof of Lemma~\ref{lemma:dnf_monotonic}, 
each transition $\langle \langle q_1, q_2 \rangle,  [a,b)\wedge [c,d) , \langle p_1, p_2 \rangle \rangle$ in the product SFA results in a formula $[\max\{a,c\},\allowbreak \min\{b,d\})$. Then, for $q_1\in Q_1$, every minimal element in the set of $q_1$'s outgoing transitions can define at most one transition in $\aut{M}$, and the same for a state $q_2\in Q_2$, and so the number of transitions from $\langle q_1, q_2 \rangle$ is at most $m_1 + m_2$, as required. 
\qed
\end{proof}

\begin{lemma}\label{lemma:complete_monotonic}
Let $\aut{M}$ be a neat SFA over a monotonic algebra. Then, transforming $\aut{M}$ into a complete SFA $\aut{M'}$ is polynomial in the size of $\aut{M}$.
\end{lemma}

\begin{proof}
In order to complete $\aut{M}$, we add a non-accepting sink $r$ in case it does not already exist, and at most $m+1$ transitions from each state $q$ to $r$, when $m$ is the out-degree of the SFA. We now prove this.  
Let $[ a, b )$ and $[ c, d )$ be two predicates labeling outgoing transitions of $q$, where
$c$ is the minimal left end-point of a predicate such that
$b<c$. Then, in order to complete $\aut{M}$, we need to add a transition to the sink, labeled by the predicate $[b, c)$. In addition, for the predicate $[a, b)$ where there is no $c>b$ that defines another predicate, if $b\neq \dmax$ then we add $[b, \dmax)$, and similarly we add $[\dmin, a)$. Then, for each state we add at most $m+1$ new transitions, 
resulting in at most $|Q|\times (m+1)$ new transitions. 
\qed
\end{proof}

\begin{definition}
For predicates over a monotonic algebra, we define a \emph{canonical representation} of a predicate $\psi$ as the simplified DNF formula which is the disjunction of all intervals satisfying $\psi$.  
\end{definition}

Note that every predicate $\psi$ over a monotonic algebra 
defines a unique partition of the domain into disjoint intervals. This unique partition corresponds to a simplified DNF formula, which is exactly the canonical representation of $\psi$. 

\begin{example}
The canonical representation of $\psi = [0, 100) \wedge( [50, 150) \vee [20, 40) )$ is $[20, 40) \vee [50, 100)$.
\end{example}

\begin{lemma}\label{lemma:canonical}
Let $\aut{M}$ be an SFA over a monotonic algebra. Then: 
\begin{enumerate}
    \item There is a unique minimal-state neat SFA $\aut{M'}$ such that $\aut{L}(\aut{M})= \aut{L}(\aut{M'})$. \label{item:unique}
    \item There is a canonical minimal-state normalized SFA $\aut{M''}$ such that $\aut{L}(\aut{M})= \aut{L}(\aut{M''})$. \label{item:canonic} 
\end{enumerate}
\end{lemma}

\begin{proof}
First, we note that for a language $\aut{L} = \aut{L}(\aut{M})$ for some SFA $\aut{M}$,
the minimal number of states in an SFA corresponds,  similarly to DFAs~\cite{Myhill57,Nerode58}, to the number of equivalence classes in the
equivalence relation $N$ defined by $(u,v)\in N \Longleftrightarrow \forall z\in \dom{D}^*: (uz\in \aut{L} \Leftrightarrow vz\in\aut{L}) $.
Indeed if $(u,v)\in N$ then there is no reason that reading them (from the initial state) should end up in different states, and if $(u,v)\notin N$ then
reading them (from the initial state) must lead to different states.

As for transitions, we have the following. 

\begin{enumerate}
    \item Let $\psi$ be a general predicate labeling a transition in $\aut{M}$. Then $\psi$ defines a unique partition of the domain into disjoint intervals, which are exactly the transitions in a neat SFA. Then, the minimal state neat SFA is unique.
    \item For normalized transitions, we can use Lemma~\ref{lemma:dnf_monotonic} to transform a general predicate labeling a transition to a DNF predicate one in linear time. A DNF predicate over a monotonic algebra is in-fact a disjunction of disjoint intervals. Then, to obtain a canonical representation, we order these intervals by order of their minimal elements.  
\end{enumerate}
\qed 
\end{proof}

\bibliographystyle{plainurl}
\bibliography{bib}

\end{document}